\documentclass[11pt]{article}
\usepackage[a4paper,top=3cm,bottom=3.5cm,left=2.5cm,right=2.5cm,marginparwidth=1.75cm]{geometry}
\usepackage[T1]{fontenc}
\usepackage{graphicx}
\usepackage[title]{appendix}
\usepackage{tikz}
\usepackage{pgfplots}
\usepackage{cutwin}

\usepackage{amsmath,yhmath}
\usepackage{amsthm}
\usepackage{amssymb}
\usepackage{hyperref}
\usepackage{cleveref}
\usepackage{epstopdf}
\usepackage{comment}
\usepackage{todonotes}
\usepackage{color}
\usepackage{float}
\usepackage[sort&compress,numbers]{natbib}
\usetikzlibrary{shapes.geometric}
\usetikzlibrary{arrows,automata,quotes}
\usepackage{iftex}
\usepackage{xfrac} 
\usepackage{amsmath,yhmath}
\usepackage{amsthm}
\usepackage{amssymb}
\usepackage{mathtools}
\usepackage{nccmath}
\usepackage{hyperref}
\usepackage{caption, subcaption}
\usepackage[title]{appendix}
\usepackage{cleveref}
\usepackage{epstopdf}
\usepackage{comment}
\usepackage{todonotes}
\usepackage{color}
\usepackage{float}
\usepackage{cutwin}
\usepackage[sort&compress,numbers]{natbib}
\usepackage{makecell}
\usepackage{array}   
\usepackage{multirow}
\usepackage{graphicx}
\usepackage{caption}
\usepackage{subcaption}
\newtheorem{thm}{Theorem}[section]

\newtheorem{ass}[thm]{Assumption}
\newtheorem{prop}[thm]{Proposition}
\theoremstyle{definition}
\newtheorem{definition}{Definition}[section]

\newtheorem{rem}[definition]{Remark}

\renewcommand{\Re}{\mathrm{Re}}

\newcommand{\ds}{\displaystyle}

\renewcommand{\S}{\mathcal{S}}
\newcommand{\Z}{\mathbb{Z}}

\numberwithin{equation}{section}
\numberwithin{figure}{section}
\title{Energy Balance and Optical Theorem for Time-Modulated Subwavelength Resonator Arrays}
\author{Erik Orvehed Hiltunen\thanks{\footnotesize Department of Mathematics, University of Oslo, Moltke Moes vei 35, 0851 Oslo, Norway (erikhilt@math.uio.no).} \and Liora Rueff\thanks{\footnotesize Department of Mathematics, ETH Z\"urich, R\"amistrasse 101, CH-8092 Z\"urich, Switzerland (liora.rueff@sam.math.ethz.ch).} }
\date{}

\begin{document}
	\maketitle
	\begin{abstract}
		We study wave propagation through a one-dimensional array of subwavelength resonators with periodically time-modulated material parameters. Focusing on a high-contrast regime, we use a scattering framework based on Fourier expansions and scattering matrix techniques to capture the interactions between an incident wave and the temporally varying system. This way, we derive a formulation of the total energy flux corresponding to time-dependent systems of resonators. We show that the total energy flux is composed of the transmitted and reflected energy fluxes, and derive an optical theorem which characterises the energy balance of the system. We provide a number of numerical experiments to investigate the impact of the time-dependency, the operating frequency and the number of resonators on the maximal attainable energy gain and energy loss. Moreover, we show the existence of lasing points, at which the total energy diverges. 
        Our results lay the foundation for the design of energy dissipative or energy amplifying systems.
	\end{abstract}
	\noindent{\textbf{Mathematics Subject Classification (MSC2000):} 35Q60, 35L05, 78A45, 78M35, 35P25} 
	\vspace{0.2cm} 
	
\section{Introduction}
    Recent advances in the design of time-dependent media, also known as Floquet metamaterials or dynamic metamaterials, have opened new avenues in wave manipulation and control. In particular, the study of resonator systems with temporally modulated material parameters has revealed promising mechanisms for non-reciprocal wave propagation \cite{jinghao_liora}, space-time wave localisation \cite{liora2024st_localisation} and enhanced scattering effects \cite{ammari2024scattering,hiltunen2024coupled}. Such systems are of interest in both physics and engineering, where precise control over wave transmission is desirable, for instance, in acoustic metamaterials and photonic devices.

    In this work, we consider a one-dimensional array of subwavelength resonators with periodically time-modulated parameters. Our aim is to analyse the resulting wave scattering phenomena, specifically its effect on the system's energy flux, in the high-contrast regime. We assume harmonic time-modulation of the material coefficients and employ a scattering formalism based on Fourier mode expansions and transfer matrix techniques building on previous work \cite{jinghao_liora, ammari2024scattering, ammari2025effectivemediumtheorytimemodulated}. This way, we develop a precise description of the total wave field and investigate the role of various parameters in shaping the energy response of the system. This study provides insight into the dynamic behaviour of time-modulated resonant structures and offers a mathematical framework for simulating and predicting their scattering characteristics. Our results have implications for the design of advanced wave-guiding devices with tunable or direction-dependent response.

    Besides a plethora of already proven effects arising from modulating the material parameters in time, we now prove how time dependency can induce energy gain or energy loss in a system. Understanding how energy evolves in dynamic media is crucial for advancing next-generation technologies in acoustics, optics and metamaterials. In particular, systems with temporally modulated material parameters can exhibit non-trivial energy exchange \cite{Milton_EnergyCons,hiltunen2025opticaltheoremgeneralizedenergy,liberal2024spatiotemporal,zhang2024conservation,nick2024wavescatteringtimeperiodiccoefficients}, enabling phenomena such as frequency conversion, amplification, dissipation and non-reciprocal transmission \cite{Cummer_SoundControl,FleurySounasAlu,1966_Holberg,ParAmplification_Koutserimpas}. These effects have direct applications in the design of acoustic diodes \cite{AcousticRectifier}, energy-harvesting structures \cite{AnisotropicMetamaterials}, wave-based computing components and time-varying photonic crystals \cite{TuringMachine}. Moreover, they support theoretical developments in wave localisation \cite{AbsenceofDiffusion} and the control of wave propagation in complex media \cite{AMMARI202017}. This work provides a rigorous analysis of energy conservation and flux balance in such systems, and identifies the mechanisms by which energy gain and loss can be engineered through modulations in time. Additionally, we present a framework in which energy conservation is guaranteed even under time-modulated material parameters. These results contribute both to the mathematical theory of wave propagation and to the practical design of tunable wave-control devices.

    Previous works, such as \cite{Milton_EnergyCons}, have already considered energy conservation in dynamic metamaterials; however, in this paper we shall consider time-modulations which are periodic in time, which induces a coupling between different frequency harmonics. We specifically consider resonator systems in the subwavelength regime, \textit{i.e.} the size of the resonators is much smaller than the operating wavelength, and assume the operating frequency to be of subwavelength order. This assumption distinguishes our results from \cite{hiltunen2025opticaltheoremgeneralizedenergy}. Moreover, we consider one-dimensional arrays of resonators which allows for more explicit calculations, since such systems only feature nearest-neighbour interactions \cite{ammari2024scattering, liora2024st_localisation}. In contrast to \cite{Milton_EnergyCons}, we aim to prove and numerically quantify energy gain and loss induced by time-dependent material parameters. Additionally, we derive explicit conditions which guarantee energy conservation despite a non-trivial temporal modulation. We show how a system of subwavelength resonators with fixed material parameters may lead to energy gain, loss or conservation solely based on the incident field illuminating the system. We choose a scattering matrix approach, which furnishes the definition of the reflected and transmitted wave field based on the incident wave field. This leads to an explicit definition of the transmitted and reflected fields, and hence total energy flux, allowing us to formulate an Optical Theorem similar to \cite{hiltunen2025opticaltheoremgeneralizedenergy}.

    We introduce the notion of maximal energy gain $P_{\mathrm{gain}}$ and maximal energy loss $P_{\mathrm{loss}}$ for a given set of parameters, where the only degree of freedom is the definition of the incident wave field. In this paper we look at how the energy gain and loss can be maximised  through the time-dependent material parameters and incident wave field's frequency. Furthermore, we assume the resonators to be small so that we can apply the asymptotic formulation of the scattering matrix, similarly to \cite{ammari2025effectivemediumtheorytimemodulated}. This sets the ground for characterising an incident wave field at which energy conservation of the system is guaranteed, regardless of the operating frequency and material parameters. These results open the door to many new design approaches of metamaterials retrieving or dissipating energy.

    We start by giving an introduction to the mathematical and physical setup of the structure in Section \ref{sec:Setting}. In Section \ref{subsec:ScattCoeffs} we define the reflected and transmitted wave field by deriving the scattering coefficients based on the scattering matrix, similarly to \cite{ammari2025effectivemediumtheorytimemodulated}. We then use this definition of the reflected and transmitted wave to derive the expression for the total energy flux, which is summarised in the Optical Theorem. Next, we investigate numerically the effect of the time-modulation, the operating frequency and the number of resonators on the total energy flux in Section \ref{sec:NumExp}. We also show the existence of lasing points in Section \ref{sec:lasingpts}, where the total energy flux diverges. Lastly, we summarise our results and conclusions in Section \ref{sec:Concl}.

\section{Problem setting}\label{sec:Setting}
    We consider a finite one-dimensional domain $\mathcal{U}$ consisting of an array of $N$ resonators $D_i:=(x_i^-,x_i^+)$ each of length $\ell_i:=x_i^+-x_i^-$, for all $i=1,\dots,N$. We denote the spacing between two resonators $D_i$ and $D_{j}$ by $\ell_{ij}:=x_{j}^--x_i^+$, for all $1\leq i<j\leq N$. We denote the area taken up by the resonators by the disjoint union $D:=\bigcup_{i=1}^ND_i$. We assume $D$ to be centred around the origin.\par 
    We denote the material parameters inside the resonator $D_i$ by $\kappa_{\mathrm{r}}\kappa_i(t)$ and $\rho_{\mathrm{r}}\rho_i(t)$. We write:
    \begin{align}
        \rho(x)=\begin{cases} 
        \rho_0, & x \notin {D}, \\
        \rho_{\mathrm{r}}\rho_i(t), & x \in D_i,
        \end{cases} \qquad \kappa(x,t)=\begin{cases} 
        \kappa_0, & x \notin {D}, \\
        \kappa_{\mathrm{r}} \kappa_i(t), & x \in D_i.
        \end{cases}
    \end{align}
    Additionally, we assume their inverse to have a finite Fourier series given by
    \begin{align}\label{def:inv_rhokappa}
        \frac{1}{\rho_i(t)}=\sum_{n=-M}^M r_{i, n} \mathrm{e}^{\mathrm{i} n \Omega t}, \quad \frac{1}{\kappa_i(t)}=\sum_{n=-M}^M k_{i, n} \mathrm{e}^{\mathrm{i} n \Omega t},\qquad\forall\,i=1,\dots,N.
    \end{align}
    Suppose that $\kappa_i$ and $\rho_i$ are periodic with frequency $\Omega$ and period $T:=2\pi/\Omega$. For our numerical simulations we shall consider
    \begin{equation}
        \rho_i(t):=\frac{1}{1+\varepsilon_{\rho,i}\cos\left(\Omega t+\phi_{\rho,i}\right)}, \quad\kappa_i(t):=\frac{1}{1+\varepsilon_{\kappa,i}\cos\left(\Omega t+\phi_{\kappa,i}\right)}, \label{eq:rho_kappa}
    \end{equation}
    for all $i=1,\dots,N$, where  $\varepsilon_{\rho,i},\,\varepsilon_{\kappa,i} \in [0,1)$ are the amplitudes of the time-modulations and $\phi_{\rho,i},\,\phi_{\kappa,i} \in [0,2\pi)$ the phase shifts. In this case we have $M=1$ and
    \begin{align}
		k_{i,-1}=\frac{\varepsilon_{\kappa,i}\mathrm{e}^{-\mathrm{i}\phi_{\kappa,i}}}{2},\quad k_{i,0}=1,\quad k_{i,1}=\frac{\varepsilon_{\kappa,i}\mathrm{e}^{\mathrm{i}\phi_{\kappa,i}}}{2},\label{def:kappa_Fcoeffs}\\
		r_{i,-1}=\frac{\varepsilon_{\rho,i}\mathrm{e}^{-\mathrm{i}\phi_{\rho,i}}}{2},\quad r_{i,0}=1,\quad r_{i,1}=\frac{\varepsilon_{\rho,i}\mathrm{e}^{\mathrm{i}\phi_{\rho,i}}}{2}.\label{def:rho_Fcoeffs}
	\end{align}
    As in \cite{ammari2024scattering}, we introduce the contrast parameter and the exterior and interior wave speeds
    \begin{equation} \label{defdelta}
        \delta:=\frac{\rho_{\mathrm{{r}}}}{\rho_0},\quad v_0:=\sqrt{\frac{\kappa_0}{\rho_0}},\quad v_{\mathrm{r}}:=\sqrt{\frac{\kappa_\mathrm{r}}{\rho_{\mathrm{r}}}},
    \end{equation}
    respectively. We assume the material to be highly contrasting, that is $\delta\ll1$, while the wave speed is of order one, \textit{i.e.} $v_0,\,v_{\mathrm{r}}=O(1)$.\par
    Let $u(x,t)$ be the total wave field propagating through the one-dimensional, time-dependent system of resonators. Note that $u$ can be divided into its scattered and incident component:
    \begin{align}\label{eq:u_total}
        u(x,t)=\begin{cases}
            u^{\mathrm{sc}}(x,t)+u^{\mathrm{in}}(x,t),&x\notin D,\\
            u^{\mathrm{sc}}(x,t),&x\in D,
        \end{cases}
    \end{align}
    where $u^{\mathrm{sc}}$ is outwards-propagating.
    
    Suppose that the incident wave impinging on the system of resonators has frequency $\omega=O(\delta^{1/2})$. Thus, we can write
    \begin{align}
        u^{\mathrm{sc}}(x,t)=\sum\limits_{n=-\infty}^{\infty}v_n^{\mathrm{sc}}(x)\mathrm{e}^{\mathrm{i}\left(\omega+n\Omega\right)t}
    \end{align}
    and $u$ is governed by the wave equation
    \begin{align}\label{eq:WaveEq}
        \frac{\partial}{\partial t}\frac{1}{\kappa(x,t)}\frac{\partial}{\partial t}u(x,t)-\frac{\partial}{\partial x} \frac{1}{\rho(x,t)}\frac{\partial}{\partial x} u(x,t)=0.
    \end{align}
    We assume the incident wave to impinge on the system of resonators from the left-hand side and to be given by
    \begin{align}\label{def:inc_wave}
        u^{\mathrm{in}}(x,t)=\begin{cases}
        \ds \sum\limits_{n=-\infty}^{\infty}a_n\mathrm{e}^{\mathrm{i}\left(k^{(n)}x+(\omega+n\Omega)t\right)}, & x< x_1^-,\\
        0, & \text{otherwise},\end{cases}
    \end{align}
    where $a_n\in\mathbb{C}$ for all $n\in\mathbb{Z}$, and $k^{(n)}:=(\omega+n\Omega)/v_0$ is the exterior wave number corresponding to each mode $n\in\mathbb{Z}$.
    
\subsection{Subwavelength resonance}
    Previous work \cite{jinghao_liora} showed the existence of subwavelength resonant frequencies for systems as considered herein. The resonant frequencies corresponding to \eqref{eq:WaveEq} are the frequencies $\omega_0$, for which there exists a non-trivial solution $u$ with a corresponding eigenmode that is essentially supported in the subwavelength frequency regime. Using the capacitance matrix \cite{jinghao_liora}, the subwavelength resonant frequencies corresponding to \eqref{eq:WaveEq} can be approximated, to order $O(\delta^{3/2})$, by the Floquet exponents $\omega_0$ of the following system of ordinary differential equations (ODEs) \cite[Theorem 3.1]{ammari2024scattering}:
    \begin{align}\label{eq:new_CapApprox}
        C\mathbf{u}(t) +\frac{1}{v_0}D\frac{\mathrm{d}}{\mathrm{d}t}\mathbf{u}(t) = -L\frac{ \rho_{\mathrm{r}}}{\delta \kappa_{\mathrm{r}}} \frac{\mathrm{d}}{\mathrm{d} t}\left(K(t)\frac{\mathrm{d} }{\mathrm{d} t}\mathbf{u}(t)\right),
    \end{align}
    where $L=\mathrm{diag}\left(\{\ell_i\}_{i=1,\dots,N}\right)$, $K(t)=\mathrm{diag}\left(\{\frac{1}{\kappa_i(t)}\}_{i=1,\dots,N}\right)$ and $C$ denotes the finite capacitance matrix whose entries are given by \cite{jinghao-silvio2023}
    \begin{align}\label{def:CapMat}
        C_{ij}=-\frac{1}{\ell_{(j-1)j}} \delta_{i(j-1)}+\left(\frac{1}{\ell_{(j-1)j}}+\frac{1}{\ell_{j(j+1)}}\right) \delta_{i j}-\frac{1}{\ell_{j(j+1)}} \delta_{i(j+1)},\quad\forall\,i,j=1,\dots,N.
    \end{align}
    The only non-zero entries of the $N\times N$ diagonal matrix $D$ are given by
    \begin{align*}
    \begin{cases}
        D_{11} = D_{NN} = 1, \quad & \text{if} \ N\geq 2,\\
        D_{11} = 2, & \text{if} \ N = 1.
    \end{cases}
    \end{align*}
    The resonant frequencies of the ODE \eqref{eq:new_CapApprox} can be computed by using a spectral method to convert the ODE into a matrix eigenvalue problem \cite{jinghao_liora}.

\subsection{Scattering coefficients}\label{subsec:ScattCoeffs}
    As in previous work \cite{jinghao_liora}, we use an exponential Ansatz so that each Fourier mode is defined through
    \begin{align}
        v_n^{\mathrm{sc}}(x)=\begin{cases}
            \alpha_n^i\mathrm{e}^{\mathrm{i}k^{(n)}x}+\beta_n^i\mathrm{e}^{-\mathrm{i}k^{(n)}x},&\forall\,x\in\left(x_{i-1}^+,x_{i}^-\right),\\
            \sum\limits_{j=-\infty}^{\infty}\left(a_j^i\mathrm{e}^{\mathrm{i}\lambda_j^ix}+b_j^i\mathrm{e}^{-\mathrm{i}\lambda_j^ix}\right)f_n^{j,i},&\forall\,x\in\left(x_i^-,x_i^+\right),
        \end{cases}
    \end{align}
    where the coefficients $\alpha_n^i,\,\beta_n^i,\,a_j^i,\,b_j^i$ need to be determined for all $i=1,\dots,N,\,j,n\in\mathbb{Z}$. The eigenpairs $\left(\lambda_j^i,\mathbf{f}^{j,i}\right)_{i\in\mathbb{Z}}$ corresponding to the $i$-th resonator can be obtained as stated in \cite[Lemma 2.1]{ammari2024scattering}. For the analytic part of this paper we will work with infinite matrices and truncate them further along for our numerical experiments in Section \ref{sec:NumExp}.
    
    We now follow the same approach as in \cite{ammari2025effectivemediumtheorytimemodulated} to define the scattered wave field through the scattering matrix.
    \begin{prop}
        Consider the resonator $D_i$, and assume that $\omega \neq m\Omega$ for all $m\in \Z$. Let the wave field on the left-hand side of $D_i$ be given by the modes $v_n(x)=\alpha^i_n\mathrm{e}^{\mathrm{i}k^{(n)}x}+\beta^i_n\mathrm{e}^{-\mathrm{i}k^{(n)}x}$ and the wave field on the right-hand side of $D_i$ be given by the modes $v_n(x)=\alpha^{i+1}_n\mathrm{e}^{\mathrm{i}k^{(n)}x}+\beta^{i+1}_n\mathrm{e}^{-\mathrm{i}k^{(n)}x}$. Then the corresponding transfer matrix is given by
        \begin{align}
            \begin{bmatrix}
                \boldsymbol{\alpha}^{i+1}\\\boldsymbol{\beta}^{i+1}
            \end{bmatrix}=\Tilde{S}^i\begin{bmatrix}
                \boldsymbol{\alpha}^{i}\\\boldsymbol{\beta}^{i}
            \end{bmatrix},\quad \Tilde{S}^i=:\begin{bmatrix}
                \Tilde{S}^i_{11} & \Tilde{S}^i_{12} \\
                \Tilde{S}^i_{21} & \Tilde{S}^i_{22}
            \end{bmatrix},
        \end{align}
        where these matrices are defined in \cite[(4.6)]{ammari2025effectivemediumtheorytimemodulated}. Furthermore, the scattering matrix is defined by
        \begin{align}\label{def:Si_tdep}
            S^i:=\begin{bmatrix}
                S^i_{11} & S^i_{12} \\ S^i_{21} & S^i_{22}
            \end{bmatrix},\quad \begin{cases}
                S^i_{11}:=\Tilde{S}^i_{11}-\Tilde{S}^i_{12}\left(\Tilde{S}^i_{22}\right)^{-1}\Tilde{S}^i_{21},\\
                S^i_{12}:=\Tilde{S}^i_{12}\left(\Tilde{S}^i_{22}\right)^{-1},\\
                S^i_{21}:=-\left(\Tilde{S}^i_{22}\right)^{-1}\Tilde{S}^i_{21},\\
                S^i_{22}:=\left(\Tilde{S}^i_{22}\right)^{-1},
            \end{cases}
        \end{align}
        and it satisfies
        \begin{align}\label{eq:Srelation}
            \begin{bmatrix}
                \boldsymbol{\alpha}^{i+1}\\\boldsymbol{\beta}^i
            \end{bmatrix}=S^i\begin{bmatrix}
                \boldsymbol{\alpha}^i\\\boldsymbol{\beta}^{i+1}
            \end{bmatrix}.
        \end{align}
    \end{prop}
    \begin{proof}
        It can be shown that the four blocks of $\Tilde{S}^i$ are explicitly given by
        {\small
        \begin{align*}
            &\Tilde{S}_{11}=\frac{1}{2}\left(\left(E_{-,\mathrm{r}}F(T^+\hat{T})+\frac{1}{\delta}E_{-,\mathrm{r}}K^{-1}F(T^+_{\lambda}\hat{T})\right)GE_{+,\mathrm{l}}+\left(\delta E_{-,\mathrm{r}}F(T^+\hat{T}_{\lambda})+E_{-,\mathrm{r}}K^{-1}F(T^+_{\lambda}\hat{T}_{\lambda})\right)GKE_{+,\mathrm{l}}\right),\\
            &\Tilde{S}_{12}=\frac{1}{2}\left(\left(E_{-,\mathrm{r}}F(T^+\hat{T})+\frac{1}{\delta}E_{-,\mathrm{r}}K^{-1}F(T^+_{\lambda}\hat{T})\right)GE_{-,\mathrm{l}}-\left(\delta E_{-,\mathrm{r}}F(T^+\hat{T}_{\lambda})+E_{-,\mathrm{r}}K^{-1}F(T^+_{\lambda}\hat{T}_{\lambda})\right)GKE_{-,\mathrm{l}}\right),\\
            &\Tilde{S}_{21}=\frac{1}{2}\left(\left(E_{+,\mathrm{r}}F(T^+\hat{T})-\frac{1}{\delta}E_{+,\mathrm{r}}K^{-1}F(T^+_{\lambda}\hat{T})\right)GE_{+,\mathrm{l}}+\left(\delta E_{+,\mathrm{r}}F(T^+\hat{T}_{\lambda})-E_{+,\mathrm{r}}K^{-1}F(T^+_{\lambda}\hat{T}_{\lambda})\right)GKE_{+,\mathrm{l}}\right),\\
            &\Tilde{S}_{22}=\frac{1}{2}\left(\left(E_{+,\mathrm{r}}F(T^+\hat{T})-\frac{1}{\delta}E_{+,\mathrm{r}}K^{-1}F(T^+_{\lambda}\hat{T})\right)GE_{-,\mathrm{l}}-\left(\delta E_{+,\mathrm{r}}F(T^+\hat{T}_{\lambda})-E_{+,\mathrm{r}}K^{-1}F(T^+_{\lambda}\hat{T}_{\lambda})\right)GKE_{-,\mathrm{l}}\right),
        \end{align*}}
        where
        \begin{align*}
            &E_{\pm,\mathrm{l}}:=\mathrm{diag}\left(\mathrm{e}^{\pm\mathrm{i}k^{(n)}x^-}\right)_{n=-\infty}^{\infty},\,E_{\pm,\mathrm{r}}:=\mathrm{diag}\left(\mathrm{e}^{\pm\mathrm{i}k^{(n)}x^+}\right)_{n=-\infty}^{\infty},\,F:=\left(f_n^{j,i}\right)_{n,j=-\infty}^\infty,\, G:=F^{-1}, \\ 
            &K:=\mathrm{diag}\left(k^{(n)}\right)_{n=-\infty}^{\infty},\,            T^{\pm}:=\mathrm{diag}\left(\left[\mathrm{e}^{\mathrm{i}\lambda_j^ix^{\pm}},\,\mathrm{e}^{-\mathrm{i}\lambda_j^ix^{\pm}}\right]\right)_{j=-\infty}^{\infty},\, T^{\pm}_{\lambda}:=\mathrm{diag}\left(\left[\lambda_j^i\mathrm{e}^{\mathrm{i}\lambda_j^ix^{\pm}},\,-\lambda_j^i\mathrm{e}^{-\mathrm{i}\lambda_j^ix^{\pm}}\right]\right)_{j=-\infty}^{\infty}.
        \end{align*}
        See \cite[Proposition 4.1]{ammari2025effectivemediumtheorytimemodulated} for the complete proof.
    \end{proof}
    For the definition of the total energy flux we need to define the wave field left and right of the system of $N$ resonators. For that we define the total scattering matrix through the total transfer matrix given by the matrix product
    \begin{align}\label{eq:totalS}
        \Tilde{S}:=\Tilde{S}^N\cdots\Tilde{S}^1,\quad \begin{bmatrix}
                \boldsymbol{\alpha}^{N+1}\\\boldsymbol{\beta}^{N+1}
            \end{bmatrix}=\Tilde{S}\begin{bmatrix}
                \boldsymbol{\alpha}^{1}\\\boldsymbol{\beta}^{1}
            \end{bmatrix}
    \end{align}
    and, thus, the total scattering matrix is given by
    \begin{align}\label{def:totalS}
        \begin{bmatrix}
                \boldsymbol{\alpha}^{N+1}\\\boldsymbol{\beta}^1
            \end{bmatrix}=S\begin{bmatrix}
                \boldsymbol{\alpha}^1\\\boldsymbol{\beta}^{N+1}
            \end{bmatrix},\quad S:=\begin{bmatrix}
            S_{11}&S_{12}\\S_{21}&S_{22}
        \end{bmatrix},\quad \begin{cases}
                S_{11}:=\Tilde{S}_{11}-\Tilde{S}_{12}\left(\Tilde{S}_{22}\right)^{-1}\Tilde{S}_{21},\\
                S_{12}:=\Tilde{S}_{12}\left(\Tilde{S}_{22}\right)^{-1},\\
                S_{21}:=-\left(\Tilde{S}_{22}\right)^{-1}\Tilde{S}_{21},\\
                S_{22}:=\left(\Tilde{S}_{22}\right)^{-1}.
            \end{cases}
    \end{align}
    We now aim to define the reflected and transmitted wave field of the system of resonators in terms of the incident wave field \eqref{def:inc_wave}, using the definition of the total scattering matrix. We can define the Fourier modes of the scattered wave field left and right of the system of resonators by
    \begin{align}
        v_n^{\mathrm{sc}}(x):=\begin{cases}
            \beta^1_n\mathrm{e}^{-\mathrm{i}k^{(n)}x},&x<x_1^-,\\
            \alpha_n^{N+1}\mathrm{e}^{\mathrm{i}k^{(n)}x},&x>x_N^+.
        \end{cases}
    \end{align}
    In the sequel, we will simplify notation and write $\boldsymbol{\alpha}:= \boldsymbol{\alpha}^1$ and $\boldsymbol{\beta}:=\boldsymbol{\beta}^{N+1}$, while the coefficients $\boldsymbol{\beta}^1:=\left[\beta^1_n\right]_{n=-\infty}^\infty$ and $\boldsymbol{\alpha}^{N+1}:=\left[\alpha^{N+1}_n\right]_{n=-\infty}^\infty$ are to be determined. Based on \eqref{def:inc_wave}, and assuming that the incident field impinges on the resonators from the left-hand side, we have 
    \begin{equation}\label{eq:incident}
    \boldsymbol{\alpha}:=\left[a_n\right]_{n=-\infty}^\infty \quad \text{and} \quad \boldsymbol{\beta}:=\boldsymbol{0}.\end{equation}
    Thus, using the notion of the total scattering matrix, we obtain
    \begin{align}
        \begin{cases}
            \boldsymbol{\alpha}^{N+1}=S_{11}\boldsymbol{\alpha},\\
            \boldsymbol{\beta}^1=S_{21}\boldsymbol{\alpha},
        \end{cases}
    \end{align}
    for a system with a rightwards propagating incident wave field. Therefore, the scattered wave field can be separated into reflected and transmitted components as follows,
    \begin{align}\label{eq:u_sc}
        u^{\mathrm{sc}}(x,t):=\begin{cases}
            u^{\mathrm{ref}}(x,t),&x<x_1^-,\\
            u^{\mathrm{tr}}(x,t),&x>x_N^+,
        \end{cases}
    \end{align}
    where
    \begin{align}
        u^{\mathrm{ref}}(x,t)=\sum\limits_{n=-\infty}^\infty\left(\sum\limits_{m=-\infty}^\infty\left(S_{11}\right)_{nm}a_m\right)\mathrm{e}^{\mathrm{i}\left(-k^{(n)}x+\left(\omega+n\Omega\right)t\right)},\label{def:ref_wave_nonasympt}\\ u^{\mathrm{tr}}(x,t)=\sum\limits_{n=-\infty}^\infty\left(\sum\limits_{m=-\infty}^\infty\left(S_{21}\right)_{nm}a_m\right)\mathrm{e}^{\mathrm{i}\left(k^{(n)}x+\left(\omega+n\Omega\right)t\right)}.\label{def:tran_wave_nonasympt}
    \end{align}

\section{Optical theorem}
    We shall now define the energy flux corresponding to the wave equation \eqref{eq:WaveEq}. Since the scattering domain is bounded, we may phrase the energy balance in terms of the inward and outward flux in the exterior region of the resonators. 
    \subsection{Total energy flux}
    Firstly, we define the energy density $E(x,t)$, which is a sum of kinetic energy and potential energy terms as
    \begin{align}
        E(x,t):=\frac{\rho_0}{2}\left(\frac{1}{\kappa(x,t)}\left|\frac{\partial u(x,t)}{\partial t}\right|^2+\frac{1}{\rho(x,t)}\left|\frac{\partial u(x,t)}{\partial x}\right|^2\right).
    \end{align}
    Here, we include the normalisation factor $\rho_0$ for later convenience. Let $I\subset\mathcal{U}$ be an open interval, then the energy over $I$ is given by
    \begin{align}
        E_I(t):=\int_I E(x,t)\,\mathrm{d}x.
    \end{align}
    We can then define the rate of change in energy over $I$ as
    \begin{align}
        P_I(t):=&\frac{\partial E_{I}}{\partial t}\nonumber\\
        =&\frac{\rho_0}{2}\int_I\left(-\frac{1}{\kappa(x,t)^2}\frac{\partial\kappa(x,t)}{\partial t}\left|\frac{\partial u(x,t)}{\partial t}\right|^2+\frac{2}{\kappa(x,t)}\Re\left(\frac{\partial \overline{u}(x,t)}{\partial t}\frac{\partial^2u(x,t)}{\partial t^2}\right)\right.\nonumber\\
        &\qquad\qquad-\left.\frac{1}{\rho(x,t)^2}\frac{\partial\rho(x,t)}{\partial t}\left|\frac{\partial u(x,t)}{\partial x}\right|^2+\frac{2}{\rho(x,t)}\Re\left(\frac{\partial \overline{u}(x,t)}{\partial x}\frac{\partial^2 u(x,t)}{\partial x\partial t}\right)\right)\,\mathrm{d}x.
    \end{align}
    Note that, in the exterior domain where the material parameters are static, the energy flux simplifies to
    \begin{align}
        P_{(a,b)}(t)=\left.\Re\left(\frac{\partial\Bar{u}(x,t)}{\partial t}\frac{\partial u(x,t)}{\partial x}\right)\right|_b^a.
    \end{align}
    At a given point $x$ outside of the resonators, we can then define the (rightward) energy flux $P(x,t)$ as    
    \begin{align}\label{def:Fstaticxt}
        P(x,t)=-\Re\left(\frac{\partial\Bar{u}(x,t)}{\partial t}\frac{\partial u(x,t)}{\partial x}\right).
    \end{align}
    Averaging over one time-period, we obtain
    \begin{align}\label{def:Fstatic}
        P(x):=\frac{1}{T}\int_0^TP(x,t)\,\mathrm{d}t.
    \end{align}
    For a rightward incident field given by \eqref{eq:incident}, we have 
    \begin{equation}
    P(x) = \begin{cases}
    P^\mathrm{in} - P^\mathrm{ref} + P^\mathrm{cross}, \quad & x< x_1^-, \\
    P^\mathrm{tr},& x> x_N^+.
    \end{cases}
    \end{equation}
    Here, $P^{\mathrm{in}}$, $P^{\mathrm{tr}}$ and $P^{\mathrm{ref}}$ are defined as in \eqref{def:Fstatic} but with $u$ replaced, respectively, by $u^\mathrm{in}$, $u^{\mathrm{tr}}$ and $u^{\mathrm{ref}}$. Moreover, $P^\mathrm{cross}$ is a cross-term resulting since $P$ depends quadratically on $u$. In  Appendix \ref{app:calc}, we show that $P^\mathrm{cross} = 0$.    Moreover, by a direct calculation, the incident, transmitted and reflected energy fluxes are given by
    \begin{align}
        P^\mathrm{in} = &||\mathcal{O}\boldsymbol{\alpha}||^2, \\  
        P^{\mathrm{tr}}=&\frac{1}{T}\int_0^T\left|\sum\limits_{n=-\infty}^\infty\sum\limits_{m=-\infty}^\infty\left(S_{21}\right)_{nm}a_m\left(\omega+n\Omega\right)\mathrm{e}^{-\mathrm{i}\frac{n\Omega}{v}(x+vt)}\right|^2\,\mathrm{d}t\nonumber\\
        =&\sum\limits_{n=-\infty}^\infty\left|\sum\limits_{m=-\infty}^\infty\left(S_{21}\right)_{nm}a_m\right|^2\left|\omega+n\Omega\right|^2\nonumber\\
        =&||\mathcal{O}S_{21}\boldsymbol{\alpha}||^2,\label{def:ptr}\\[10pt]
        P^{\mathrm{ref}}=&\frac{1}{T}\int_0^T\left|\sum\limits_{n=-\infty}^\infty\sum\limits_{m=-\infty}^\infty\left(S_{11}\right)_{nm}a_m\left(\omega+n\Omega\right)\mathrm{e}^{-\mathrm{i}\frac{n\Omega}{v}(x-vt)}\right|^2\,\mathrm{d}t\nonumber\\
        =&\sum\limits_{n=-\infty}^\infty\left|\sum\limits_{m=-\infty}^\infty\left(S_{11}\right)_{nm}a_m\right|^2\left|\omega+n\Omega\right|^2\nonumber\\
        &=||\mathcal{O}S_{11}\boldsymbol{\alpha}||^2,\label{def:pref}
    \end{align}
    where $\left|\left|\cdot\right|\right|$ denotes the $\ell^2$-norm and $\mathcal{O}:=\mathrm{diag}\left(\omega+n\Omega\right)_{n=-\infty}^\infty$.

    By comparing the energy flux on the left and right sides of the scattering domain, we can now determine the energy balance of the system. This now allows us to formulate the following \textit{Optical Theorem} for one-dimensional, time-dependent systems of subwavelength resonators.
    \begin{thm}[Optical Theorem]\label{thm:Optical}
        We consider a finite, time-periodic resonator system with scattering matrix given by \eqref{eq:totalS}, and a rightward incident field given by \eqref{eq:incident}.
        Then, the transmitted and reflected energy fluxes are given by
        \begin{align}
            P^{\mathrm{tr}}=||\mathcal{O}S_{21}\boldsymbol{\alpha}||^2,\quad P^{\mathrm{ref}}=||\mathcal{O}S_{11}\boldsymbol{\alpha}||^2,
        \end{align}
        respectively, where $\mathcal{O}:=\mathrm{diag}\left(\omega+n\Omega\right)_{n=-\infty}^\infty$. Moreover, the total energy flux is given by
        \begin{align}P^\mathrm{tot}:=P^{\mathrm{tr}}+P^{\mathrm{ref}}=||\mathcal{O}S_{11}\boldsymbol{\alpha}||^2+||\mathcal{O}S_{21}\boldsymbol{\alpha}||^2.
        \end{align}
        We define three possible energy states:
        \begin{itemize}
            \item Energy gain: $P^\mathrm{tot}>||\mathcal{O}\boldsymbol{\alpha}||^2$;
            \item Energy conservation: $P^\mathrm{tot}=||\mathcal{O}\boldsymbol{\alpha}||^2$;
            \item Energy loss: $P^\mathrm{tot}<||\mathcal{O}\boldsymbol{\alpha}||^2$.
        \end{itemize}
    \end{thm}

    \subsection{Point approximation}
    We now consider the case of small resonators, \textit{i.e.} we let $\ell\to0$, which sets the ground for more explicit results. First, we prove an asymptotic expression of the total scattering matrix \eqref{def:totalS}, similarly as in \cite[Lemma 4.2]{ammari2025effectivemediumtheorytimemodulated}, for which we make the following assumption on the scaling of the material parameters:
    \begin{ass}
        Assume that each resonator is of equal length $\ell_i\equiv\ell$. We assume that there exists a set of parameters $\gamma,\,\mu,\,\xi>0$ independent of $\ell$ such that $\delta=\gamma\ell^2$ and since we assume the operating frequency to be at the subwavelength resonant frequency we set $\omega=\mu\ell$ and $\Omega=\xi\ell$ \cite{jinghao_liora}. As a consequence, the values $\lambda_j^i$ are of the same order, \textit{i.e.} there exists $c_j^i>0$ such that $\lambda_j^i=c_j^i\ell$, for all $i=1,\dots,N,\,j\in\mathbb{Z}$. In order to carry out analytical calculations, we shall assume these scalings for the remainder of this paper.
    \end{ass}
    \begin{thm}\label{thm:asymptoticS}
        As $\ell\to0$, the total scattering matrix is of the following form:
        \begin{align}\label{def:asymptotic_Si_tdep}
            S=\begin{bmatrix}
                I_{\infty} & 0 \\ 0 & I_{\infty}
            \end{bmatrix}+\begin{bmatrix}
                \mathcal{G} & 0 \\ 0 & \mathcal{G}
            \end{bmatrix}\begin{bmatrix}
                I_{\infty} & I_{\infty} \\ I_{\infty} & I_{\infty}
            \end{bmatrix} + O(\ell^2),
        \end{align}
        where $I_{\infty}$ denotes the infinite identity matrix, and $\mathcal{G}:=\left(g_{nm}\right)_{n,m=-\infty}^\infty$ is given by
        \begin{align}\label{def:g_Acal_tdep}
            \mathcal{G}:=\left(I_{\infty}-\mathcal{A}\right)^{-1}\mathcal{A},\quad \mathcal{A}:=\sum\limits_{i=1}^N\frac{\mathrm{i}\ell}{2\gamma}\mathrm{diag}\left(\frac{v_0}{\mu+n\xi}\right)_{n=-\infty}^{\infty}F^i\mathrm{diag}\left((c_j^i)^2\right)_{j=-\infty}^{\infty}G^i,
        \end{align}
        where $F^i:=\left(\boldsymbol{f}^{j,i}\right)_{j=-\infty}^{\infty}$, namely, the columns of $F^i$ are the eigenvectors $\boldsymbol{f}^{j,i}$, and we write $G^i:=\left(F^i\right)^{-1}$.
    \end{thm}
    \begin{proof}
        The formula follows directly by induction over $N$ using the methods from the proof of \cite[Lemma 4.2]{ammari2025effectivemediumtheorytimemodulated}.
    \end{proof}
    Using the notion of the asymptotic expression of the total scattering matrix and the asymptotic result proven in Theorem \ref{thm:asymptoticS}, we obtain
    \begin{align}
        \begin{cases}
            \boldsymbol{\alpha}^{N+1}=\left(I_{\infty}+\mathcal{G}\right)\boldsymbol{\alpha}+O(\ell^2),\\
            \boldsymbol{\beta}^1=\mathcal{G}\boldsymbol{\alpha}+O(\ell^2).
        \end{cases}
    \end{align}
    Therefore, we may define the reflected and scattered wave field, for $\ell$ sufficiently small, by
    \begin{align}
        u^{\mathrm{ref}}(x,t)=\sum\limits_{n=-\infty}^\infty\left(\sum\limits_{m=-\infty}^\infty g_{nm}a_m\right)\mathrm{e}^{\mathrm{i}\left(-k^{(n)}x+\left(\omega+n\Omega\right)t\right)},\label{def:ref_wave}\\ u^{\mathrm{tr}}(x,t)=\sum\limits_{n=-\infty}^\infty\left(a_n+\sum\limits_{m=-\infty}^\infty g_{nm}a_m\right)\mathrm{e}^{\mathrm{i}\left(k^{(n)}x+\left(\omega+n\Omega\right)t\right)},\label{def:tran_wave}
    \end{align}
    respectively. This allows us to formulate the Optical Theorem for small resonators. 
    \begin{thm}[Optical Theorem for Small Resonators]\label{thm:Optical_asympt}
    For small $\ell$, the transmitted and reflected energy fluxes are given by
        \begin{align}
            P^{\mathrm{tr}}=||\mathcal{O}(I_{\infty}+\mathcal{G})\boldsymbol{\alpha}||^2,\quad P^{\mathrm{ref}}=||\mathcal{O}\mathcal{G}\boldsymbol{\alpha}||^2,
        \end{align}
        up to errors of order $O(\ell^2)$, where $\mathcal{O}:=\mathrm{diag}\left(\omega+n\Omega\right)_{n=-\infty}^\infty$. Hence, the total energy flux is given by
        \begin{align}
            P^\mathrm{tot}=||\mathcal{O}(I_{\infty}+\mathcal{G})\boldsymbol{\alpha}||^2+||\mathcal{O}\mathcal{G}\boldsymbol{\alpha}||^2.
        \end{align}
    \end{thm}
    \begin{rem}
        For the remainder of this paper we shall assume $||\mathcal{O}\boldsymbol{\alpha}||=1$.
    \end{rem}
    In Figure \ref{fig:E_overview} we plot the total energy flux as a function of the incident wave field's coefficients $a_0$ and $a_{\pm1}$. These numerical results allow us to clearly identify the three areas as formulated in the Optical Theorem. Additionally, we mark the maxima and minima of the surface in Figure \ref{fig:E_overview}, which correspond to the highest attainable energy gain and loss, respectively. We shall further investigate these two values numerically in Section \ref{sec:NumExp}.\par 
    \begin{figure}[H]
        \centering
        \includegraphics[width=0.85\textwidth]{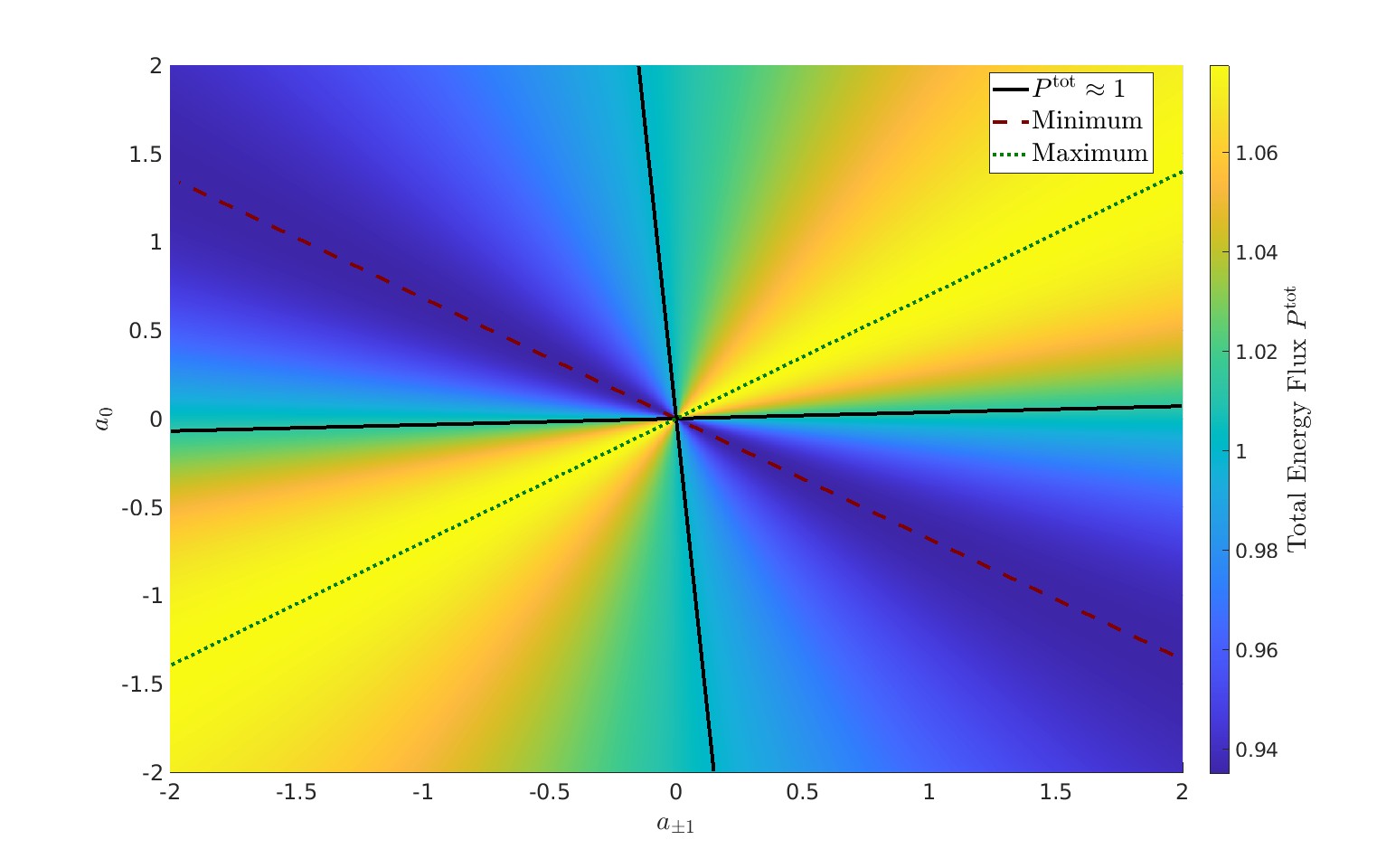}
        \caption{The total energy flux $P^\mathrm{tot}$ for a single resonator, where the incident wave has three modes given by $a_{-1}\mathrm{e}^{\mathrm{i}k^{(-1)}x}$, $a_{0}\mathrm{e}^{\mathrm{i}k^{(0)}x}$ and $a_{1}\mathrm{e}^{\mathrm{i}k^{(1)}x}$ with $a_{-1}=a_1$ and normalised such that $||\mathcal{O}\boldsymbol{\alpha}||=1$.  Here, we compute the total energy flux $P^\mathrm{tot}$ as a function of $a_0$ and $a_{\pm1}$, and observe regions of energy conservation, energy gain and energy loss. We have used  the following material parameters: $\ell=0.1$, $\varepsilon_{\rho}=0,\,\varepsilon_{\kappa}=0.9,\,\mu=0.9,v_{\mathrm{r}}=v_0=1,\,\gamma=0.05,\,\xi=0.2,\,\phi_{\kappa}=\pi/2$.}\label{fig:E_overview}
    \end{figure}
    The following result underlines the relevance of the above proven asymptotic result, since it provides an explicit condition to guarantee energy conservation, regardless of the material parameters.
    \begin{thm}\label{thm:Econs_Geigen}
    Consider a scattering system with a single resonator, \textit{i.e.} $N=1$.
        Let $\left(\varphi,\boldsymbol{v}\right)$ be an eigenpair of the matrix $\mathcal{G}$ defined in \eqref{def:g_Acal_tdep} such that $||\mathcal{O}\boldsymbol{v}||=1$. If the incident field is given by $\boldsymbol{\alpha}=\boldsymbol{v}$, then the energy is conserved, \textit{i.e.} $P^\mathrm{tot}=1$.
    \end{thm}
    \begin{proof}
        If $\boldsymbol{\alpha}=\boldsymbol{v}$, then 
        \begin{align*}
            P^\mathrm{tot}=||\mathcal{O}(I_{2K+1}+\mathcal{G})\boldsymbol{v}||^2+||\mathcal{O}\mathcal{G}\boldsymbol{v}||^2=2|\varphi|^2+2\mathrm{Re}(\varphi)+1.
        \end{align*}
        Therefore, we obtain energy conservation if and only if $2|\varphi|^2+2\mathrm{Re}(\varphi)+1=1$, which we shall prove to be the case here.\par 
        Recall that $\mathcal{G}=\left(I_{\infty}-\mathcal{A}\right)^{-1}\mathcal{A}$ and let $\lambda$ be the eigenvalue of $\mathcal{A}$ such that $\varphi=\lambda/(1-\lambda)$. The function $f(z):=z/(1-z)$ is the Möbius transform, mapping the imaginary axis onto the circle of radius $1/2$ and centred at $-1/2$ in the complex plane. Thus, if $z\in\mathbb{C}\setminus\mathbb{R}$, then $f(z)=-0.5+0.5\mathrm{e}^{\mathrm{i}\theta}$, for some $\theta\in[0,2\pi)$. If we assume that $\lambda$ is purely imaginary, it follows that $\varphi=-0.5+0.5\mathrm{e}^{\mathrm{i}\theta}$, which is equivalent to 
        \begin{align}\label{eq:cons}
            2|\varphi|^2+2\mathrm{Re}\left(\varphi\right)+1=1,
        \end{align}
        which furnishes energy conservation. Therefore, if we can prove that the eigenvalues of $\mathcal{A}$ are purely imaginary, the claim $P^\mathrm{tot}=1$ follows.\par
        Assume $N=1$, then the matrix $\mathcal{A}$ is given by
        \begin{align}
            \mathcal{A}:=\frac{\mathrm{i}\ell}{2\gamma}\mathrm{diag}\left(\frac{v}{\mu+n\xi}\right)_{n=-\infty}^{\infty}F\mathrm{diag}\left((c_j)^2\right)_{j=-\infty}^{\infty}G,
        \end{align}
        where $\gamma,\,v,\,\mu,\,\xi,\,c_j$ are non-zero and real. We now aim to prove that $\mathcal{A}$ is similar to a Hermitian matrix multiplied by the imaginary unit $\mathrm{i}$, which will imply that the eigenvalues of $\mathcal{A}$ are purely imaginary.  Firstly, we note that
        \begin{align}
            \mathrm{diag}\left(\frac{v}{\mu+n\xi}\right)_{n=-\infty}^{\infty}=\ell v\mathcal{O}^{-1}.
        \end{align}
        Next, we recall that $\lambda_j=c_j\ell$ are the square-roots of the eigenvalues of a matrix $C_1$ and the columns of $F$ are its corresponding eigenvectors \cite[Lemma III.3]{jinghao_liora}. Note that by definition of this matrix $C_1$, we may write $C_1=\frac{1}{\left(v_{\mathrm{r}}\right)^2}R_1^{-1}\mathcal{O}K_1\mathcal{O}$, where
        \begin{align}
            R_1:=\begin{bmatrix}
                \ddots & \ddots & \ddots & & \\
                & r_{1,-1} & r_{1,0} & r_{1,1} &\\
                & & \ddots & \ddots & \ddots
            \end{bmatrix},\quad K_1:=\begin{bmatrix}
                \ddots & \ddots & \ddots & & \\
                & k_{1,-1} & k_{1,0} & k_{1,1} &\\
                & & \ddots & \ddots & \ddots
            \end{bmatrix},
        \end{align}
        for the Fourier coefficients of the inverse of $\rho$ and $\kappa$ as given in \eqref{def:inv_rhokappa}, respectively. Thus, since $G=F^{-1}$, we arrive at
        \begin{align}
            F\mathrm{diag}\left((c_j)^2\right)_{j=-\infty}^{\infty}G=\frac{1}{\ell^2\left(v_{\mathrm{r}}\right)^2}R_1^{-1}\mathcal{O}K_1\mathcal{O},
        \end{align}
        and therefore,
        \begin{align}
            \mathcal{A}=\frac{\mathrm{i}v}{2\gamma\left(v_{\mathrm{r}}\right)^2}\mathcal{O}^{-1}R_1^{-1}\mathcal{O}K_1\mathcal{O}.
        \end{align}
        We note that the matrices $R_1$ and $K_1$ commute, and since they are Hermitian and positive definite, their square-roots exist. Let $Q:=K_1^{1/2}R_1^{1/2}\mathcal{O}$, then
        \begin{align}
            Q\mathcal{A}Q^{-1}=\frac{\mathrm{i}v}{2\gamma\left(v_{\mathrm{r}}\right)^2}K_1^{1/2}R_1^{-1/2}\mathcal{O}R_1^{-1/2}K_1^{1/2},
        \end{align}
        which (up to the factor of $\mathrm{i}$) is a symmetric matrix. Hence, the eigenvalues of $K_1^{1/2}R_1^{-1/2}\mathcal{O}R_1^{-1/2}K_1^{1/2}$ are real and, thus, the eigenvalues of $Q\mathcal{A}Q^{-1}$ are purely imaginary, which concludes the proof.
    \end{proof}
    \begin{rem}
        The proof of Theorem \ref{thm:Econs_Geigen} holds true for a single resonator, but can be trivially generalised to $N>1$, if each resonator's material parameters are defined such that the eigenvalues and eigenvectors of $C_i$ are independent of $i$. Specifically, this is the case if the material parameters are identical across all resonators. 
    \end{rem}
    The numerical results depicted in Figure \ref{fig:Geigenvals} support the proof of Theorem \ref{thm:Econs_Geigen}. Additionally, we show in Figure \ref{fig:Etot_Geigenvals} that the energy is indeed conserved along the eigenvectors. Namely, we plot the total energy flux resulting from an incident wave field with coefficients
    \begin{align}
        \boldsymbol{\alpha}=\frac{c_1\boldsymbol{v}_1+c_2\boldsymbol{v}_2}{||\mathcal{O}\left(c_1\boldsymbol{v}_1+c_2\boldsymbol{v}_2\right)||}
    \end{align}
    and see that we achieve energy conservation for $c_1=0$ or $c_2=0$.\par 
    To summarise, if one wants to guarantee energy conservation in a time-modulated system, the coefficients $\boldsymbol{\alpha}$ of the incident wave should be an eigenvector of the matrix $\mathcal{G}$. 
    \begin{figure}[H]
        \begin{subfigure}{0.48\textwidth}
            \centering
            \includegraphics[width=0.99\textwidth]{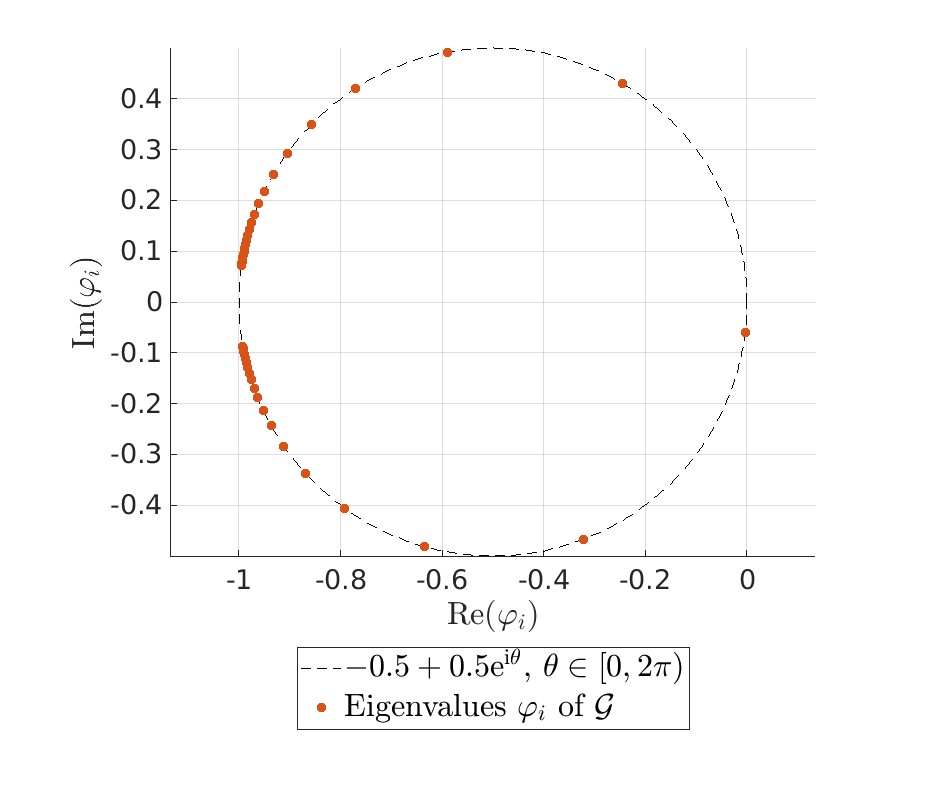}
            \caption{The eigenvalues $\left(\varphi_i\right)_{i=-K}^K$ of $\mathcal{G}$ for truncation number $K=20$.}\label{fig:Geigenvals}
        \end{subfigure}
        \hfill
        \begin{subfigure}{0.48\textwidth}
            \centering
            \includegraphics[width=0.99\textwidth]{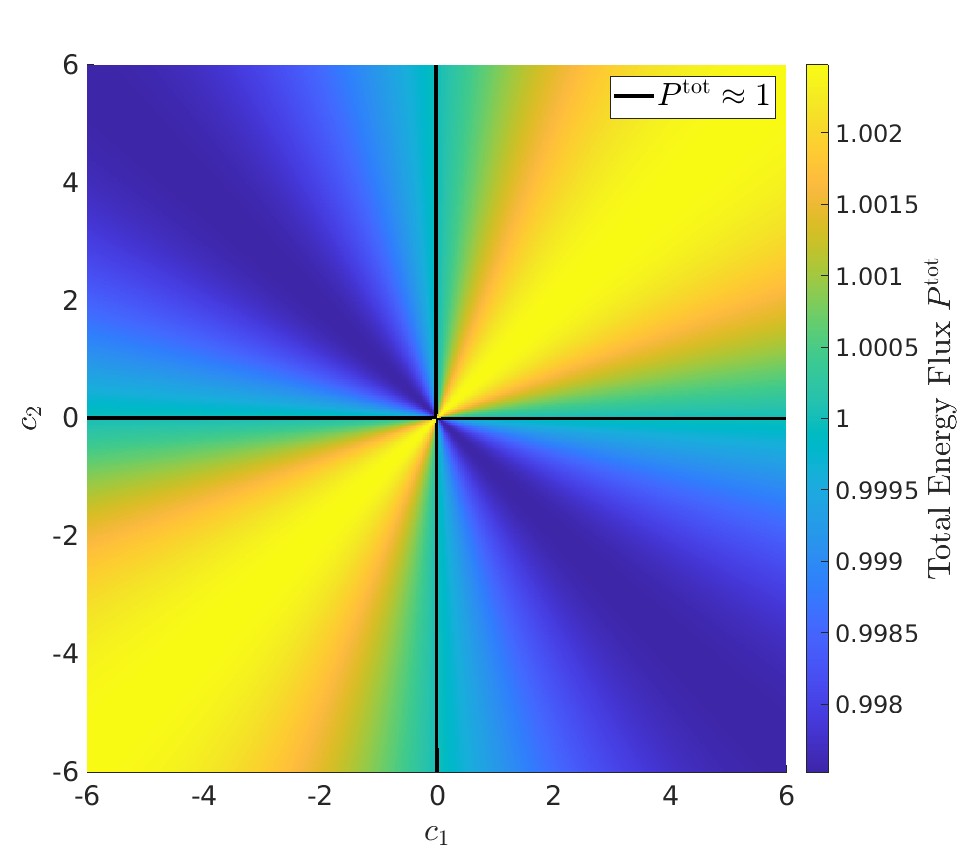}
            \caption{The total energy flux for incident wave field with coefficients $c_1\boldsymbol{v}_1+c_2\boldsymbol{v}_2$ upon normalisation, where $\boldsymbol{v}_1$ and $\boldsymbol{v}_2$ are eigenvectors of $\mathcal{G}$.}\label{fig:Etot_Geigenvals}
        \end{subfigure}
        \caption{The energy is conserved if the incident field is given by an eigenvector of $\mathcal{G}$. In (a) we see that all eigenvalues of $\mathcal{G}$ lie on the circle centred at $-\frac{1}{2}$ and with radius $\frac{1}{2}$, in agreement with \eqref{eq:cons}. In (b), we repeat the energy phase diagram of \Cref{fig:E_overview}, but where the coordinate axes are now given by the first two eigenvectors. As expected, we observe energy conservation along these axes. The material parameters are set to be $N=3,\,\ell=0.1,\,\delta=5\times10^{-4},\,v_{\mathrm{r}}=v_0=1,\,\omega=0.04,\,\Omega=0.1,\,\varepsilon_{\kappa}=\varepsilon_{\rho}=0.5,\,\phi_{\kappa}=\phi_{\rho}=\pi/2$.}
    \end{figure}

\section{Numerical experiments}\label{sec:NumExp}
Next, we provide a number of numerical illustrations to illustrate the energy balance of the system. In this section, we shall assume $\ell\ll1$ for the sake of convenience, but we note that our approach also works for a general $\ell>0$.

The wave field $u(x,t)$ has infinitely many modes $v_n(x)$, which we approximate by truncated Fourier series, with $K\in\mathbb{N}$ sufficiently large \cite{jinghao_liora}:
        \begin{align}\label{eq:utrunc}
            u^{\mathrm{sc}}(x,t)\approx\sum\limits_{n=-K}^{K}v_n^{\mathrm{sc}}(x)\mathrm{e}^{\mathrm{i}(\omega+n\Omega )t},\quad u^{\mathrm{in}}(x,t)\approx\sum\limits_{n=-K}^{K}a_n\mathrm{e}^{\mathrm{i}\left(k^{(n)}x+(\omega+n\Omega)t\right)}.
        \end{align}
        As a result we also need to truncate the scattering matrix to be of size $2(2K+1)\times2(2K+1)$ and $\mathcal{G}$ is of size $(2K+1)\times(2K+1)$\par

    \subsection{Energy dependency on material parameters}   
    For a given setting we shall denote the maximal achievable energy gain by $P_{\mathrm{gain}}$ and the maximal achievable energy loss by $P_{\mathrm{loss}}$. Note that $P_{\mathrm{gain}}$ and $P_{\mathrm{loss}}$ are the maxima and minima, respectively, of the surface shown in Figure \ref{fig:E_overview}. We now aim to better understand the impact of the various material parameters on $P_{\mathrm{gain}}$ and $P_{\mathrm{loss}}$.

    For the sake of simplicity we assume the incident wave field to be given by
    \begin{align}\label{def:incfield}
        u^{\mathrm{in}}(x,t)=a_{-1}\mathrm{e}^{\mathrm{i}\left(k^{(-1)}x+(\omega-\Omega)t\right)}+a_{0}\mathrm{e}^{\mathrm{i}\left(k^{(0)}x+\omega t\right)}+a_{1}\mathrm{e}^{\mathrm{i}\left(k^{(1)}x+(\omega+\Omega)t\right)},\quad a_{-1}=a_1.
    \end{align}
    We regard $P_{\mathrm{gain}}$ and $P_{\mathrm{loss}}$ as functions of $\left(a_0,a_{\pm1}\right)$, whereas, for a given set of material parameters we determine $\left(a_0,a_{\pm1}\right)$ for which $P^\mathrm{tot}$ attains its maximum and minimum value, respectively.

    Firstly, we see that for increasing $\varepsilon_{\kappa}$ we can attain higher energy gain and loss, as depicted in Figure \ref{fig:Elossgain_epsilonkappa}, and the same holds true for $\varepsilon_{\rho}$, which can be seen in Figure \ref{fig:Elossgain_epsilonrho}. As expected and proven in previous work \cite{feppon_1d}, for static material parameters, \textit{i.e.} $\varepsilon_{\kappa}=\varepsilon_{\rho}=0$, a system of $N$ resonators obeys energy conservation. The conclusions from the results shown in Figures \ref{fig:Elossgain_epsilonkappa} and \ref{fig:Elossgain_epsilonrho} can be shown numerically to also hold true for larger $N$.
    \begin{figure}[H]
        \centering
        \includegraphics[width=0.6\textwidth]{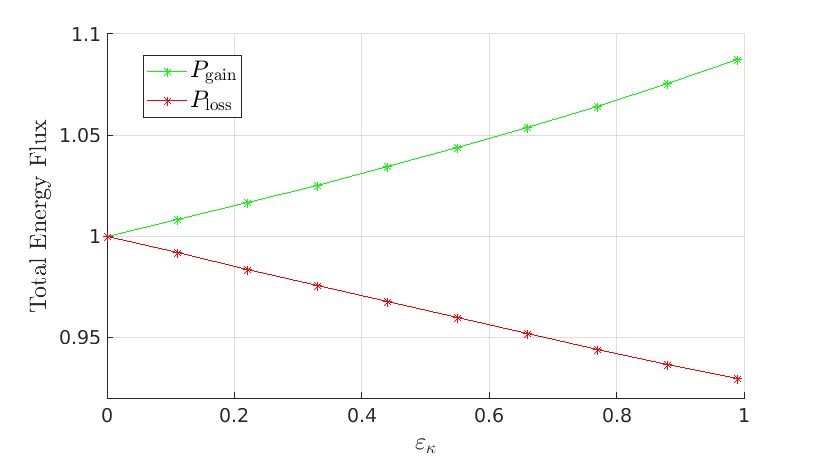}
        \caption{The maximal attainable energy gain and energy loss for a single resonator of length $\ell=0.1$ with the following material parameters: $\varepsilon_{\rho}=0,\,\mu=0.9,v_{\mathrm{r}}=v_0=1,\,\gamma=0.05,\,\xi=0.2,\,\phi_{\kappa}=\pi/2$.}\label{fig:Elossgain_epsilonkappa}
    \end{figure}
    \begin{figure}[H]
        \centering
        \includegraphics[width=0.6\textwidth]{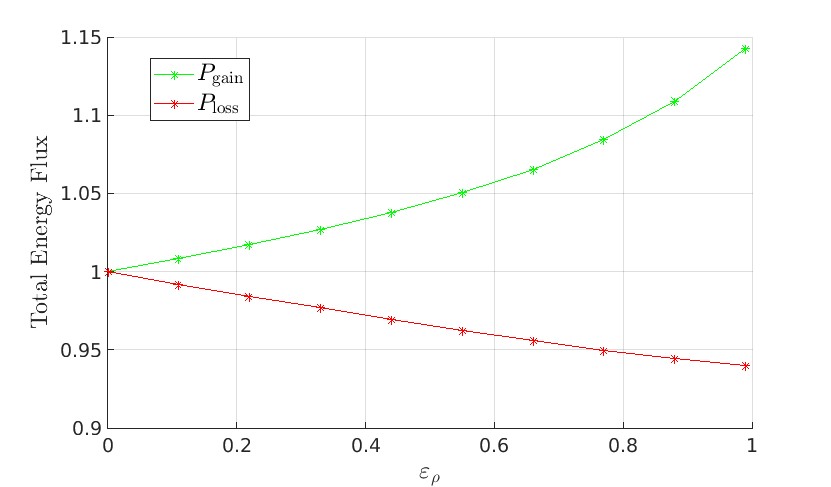}
        \caption{The maximal attainable energy gain and energy loss for a single resonator of length $\ell=0.1$ with the following material parameters: $\varepsilon_{\kappa}=0,\,\mu=0.9,v_{\mathrm{r}}=v_0=1,\,\gamma=0.05,\,\xi=0.2,\,\phi_{\rho}=\pi/2$.}\label{fig:Elossgain_epsilonrho}
    \end{figure}

    In Figure \ref{fig:Elossgain_omega} we investigate how the operating frequency $\omega$ can impact the maximum attainable energy gain and loss. Specifically, we see that for $\omega$ close to the resonant frequency $\omega_0$, which is defined through \eqref{eq:new_CapApprox}, the energy gain and loss can be maximally amplified with the suitable choice of the incident wave field's parameters $a_0$ and $a_{\pm1}$.\par 
    The results in Figure \ref{fig:Elossgain_omega} clearly show that for a wave field with operating frequency $\omega$ near the resonant frequency $\omega_0$, the energy gain or loss can be maximised. Combining the results from Figure \ref{fig:Elossgain_epsilonkappa} and Figure \ref{fig:Elossgain_omega} allows us to conclude that energy gain and loss can be amplified upon modulating the material parameters in time and choosing the operation frequency to be at a resonant frequency $\omega_0$, as shown in Figure \ref{fig:Elossgain_kappaomega0}. The same result holds true for $\varepsilon_{\rho}$.
    \begin{figure}[H]
        \centering
        \includegraphics[width=0.6\textwidth]{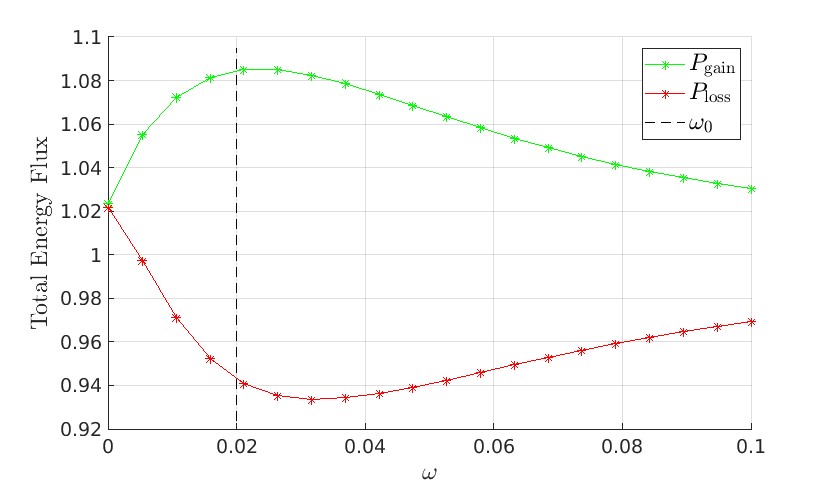}
        \caption{The maximal attainable energy gain and energy loss for a single resonator of length $\ell=0.1$ with the following material parameters: $\varepsilon_{\rho}=0,\,\varepsilon_{\kappa}=0.5,\,v_{\mathrm{r}}=v_0=1,\,\gamma=0.05,\,\xi=0.2,\,\phi_{\kappa}=\pi/2$.}\label{fig:Elossgain_omega}
    \end{figure}
    \begin{figure}[H]
        \centering
        \includegraphics[width=0.6\textwidth]{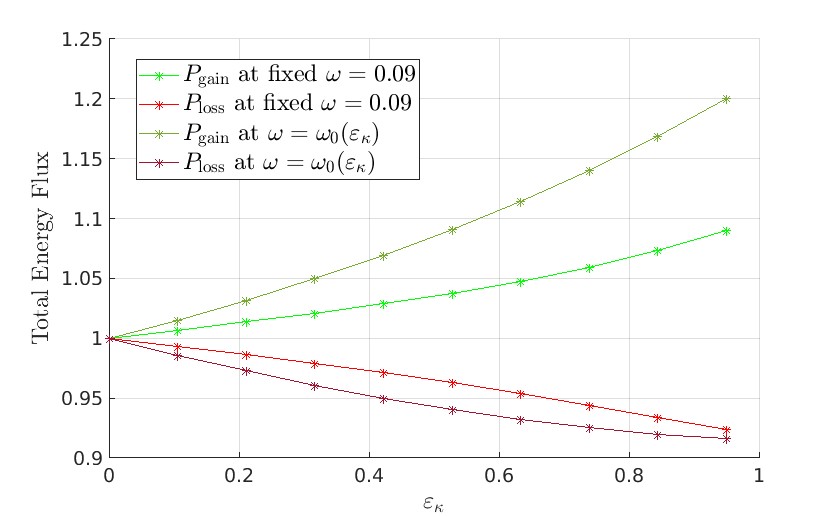}
        \caption{The maximal attainable energy gain and energy loss for two resonators each of length $\ell=0.1$ with the following material parameters: $\varepsilon_{\rho}=0,\,v_{\mathrm{r}}=v_0=1,\,\gamma=0.05,\,\xi=0.2,\,\phi_{\kappa}=\pi/2$.}\label{fig:Elossgain_kappaomega0}
    \end{figure}\par 
    \begin{figure}[H]
        \centering
        \includegraphics[width=0.6\textwidth]{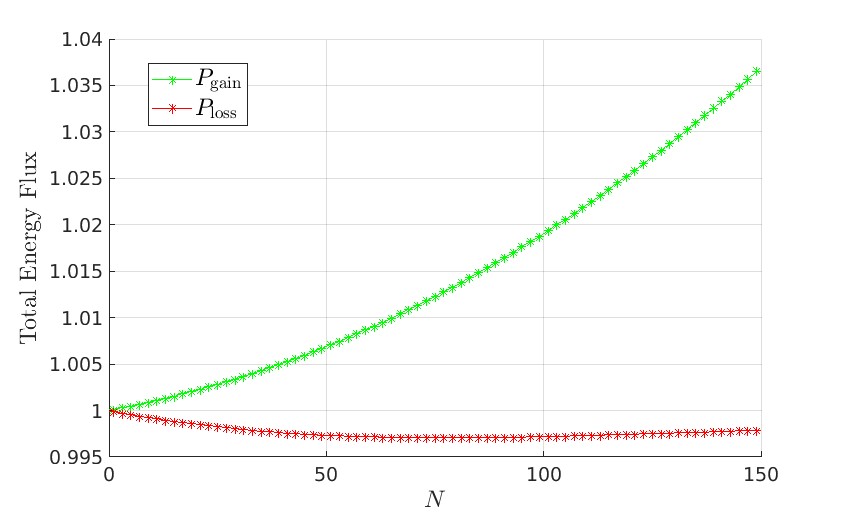}
        \caption{The maximal attainable energy gain and energy loss for $N$ resonators each of length $\ell=0.01$ with the following material parameters: $\varepsilon_{\rho}=0,\,\varepsilon_{\kappa}=0.8,\,\mu=0.2,\,v_{\mathrm{r}}=v_0=1,\,\gamma=1,\,\xi=0.2,\,\phi_{\kappa}=\pi/2$.}\label{fig:Elossgain_N}
    \end{figure}
    Next, we aim to understand how the number of resonators impacts the total energy flux. For that, we consider a system of $N$ resonators and compute the corresponding resonant frequencies for a given set of parameters. We note that for any $N$ with parameter choice as in Figure \ref{fig:Elossgain_N}, $\omega_0=0.02$ is always a resonant frequency. Therefore, we set $\omega=0.02$ in Figure \ref{fig:Elossgain_N}.\par 
    From Figure \ref{fig:Elossgain_N} we see that for increasingly many resonators, \textit{i.e.} as $N\to\infty$, the maximum possible energy gain grows continuously. This effect is further exemplified in the next section when determining the lasing points. We note that for the herein chosen fixed parameters, the maximal energy loss does not depend heavily on $N$.

\subsection{Lasing points}\label{sec:lasingpts}
    We now aim to identify the \textit{lasing points} corresponding to the capacitance matrix approximation formula \eqref{eq:new_CapApprox}. A lasing point corresponds to a set of parameter values at which lasing, \textit{i.e.} wave amplification by stimulated emission of radiation, occurs. Equivalently, a lasing point is given by a set of parameter values at which a resonance, characterised by \eqref{eq:new_CapApprox}, has vanishing imaginary part. The characterisation of the scattered field as derived in \cite[eq. (4.6)]{ammari2024scattering} makes it apparent that, since $\omega\in\mathbb{R}$, the scattered field diverges at a lasing point, hence, the total energy also diverges.\par 
    \begin{figure}[H]
        \begin{subfigure}{0.88\textwidth}
            \centering
            \includegraphics[width=0.98\textwidth]{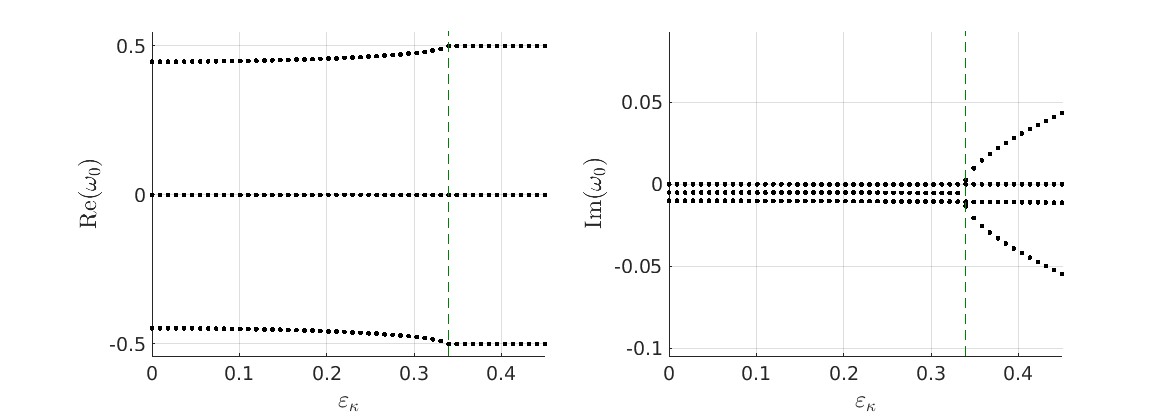}
            \caption{{\footnotesize The subwavelength resonant frequencies $\omega_0$ as a function of $\varepsilon_{\kappa}$. The green dashed line marks the lasing point.}}
            \label{fig:omega_lasing}
        \end{subfigure}
        \centering
        \begin{subfigure}{0.88\textwidth}
            \centering
            \includegraphics[width=0.98\textwidth]{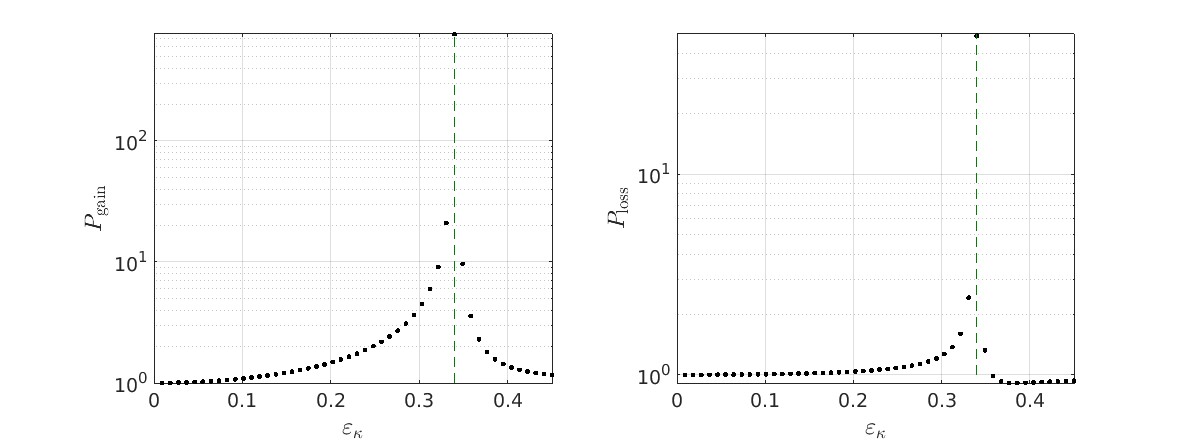}
            \caption{{\footnotesize The maximal possible energy gain and energy loss, where the green dashed line marks the lasing point. Note that we display our results here with a \texttt{semilogy} scale.}}
            \label{fig:E_lasing}
        \end{subfigure}
        \caption{A system of $N=2$ resonators each of length $\ell=0.1$ with $\delta=10^{-3},\,v_{\mathrm{r}}=v_0=1,\,\Omega=1,\,\varepsilon_{\rho}=0,\,\phi_{\kappa}=\pi/2$.}
        \label{fig:lasing}
    \end{figure}
    The results shown in Figure \ref{fig:lasing} clearly support this statement. For a given set of material parameters, we see in Figure \ref{fig:omega_lasing} that there exists a lasing point around $\varepsilon_{\kappa}\approx0.3233$ at which the total energy diverges, as shown in Figure \ref{fig:E_lasing}. In particular, $P_{\mathrm{loss}}>1$ which means that regardless of the incident wave field, it is impossible to achieve a dissipative system.

\section{Concluding remarks}\label{sec:Concl}
    By studying the total energy of a system of time-modulated subwavelength resonators upon a given incident wave field, this work presents a natural continuation of \cite{ammari2024scattering}, which investigates the scattered wave field of a time-modulated metamaterial. Using a scattering matrix approach, we were able to define the transmitted and reflected wave field linearly in terms of the given incident wave field. Moreover, we assumed the resonators to be small, \textit{i.e.} $\ell\to0$, to enable more explicit expressions. This allowed us to capture the effects of temporal modulation on the energy flux and quantify the transmission and reflection of energy in dynamic media.

    We provided a detailed derivation of the total energy flux and established an Optical Theorem adapted to time-dependent systems, distinguishing between energy gain, energy loss and energy conservation, see Theorem \ref{thm:Optical}. Crucially, we showed that time-modulated systems can be engineered to amplify or dissipate energy depending on the choice of the modulation parameters $\varepsilon_{\kappa}$ and $\varepsilon_{\rho}$, the operating frequency $\omega$ and incident wave field $\boldsymbol{\alpha}$. Our analysis identified the coupling between Fourier modes as the key mechanism enabling non-trivial energy dynamics.

    We defined the values $P_{\mathrm{gain}}$ and $P_{\mathrm{loss}}$ to be the maximal attainable energy gain and loss, respectively, for a fixed set of parameters where the only degree of freedom is the incident wave field's coefficients $\boldsymbol{\alpha}$. We demonstrated that maximal energy gain or loss occurs at operating frequencies near subwavelength resonant frequencies, and that these effects are significantly enhanced in systems with many resonators and strong time-modulation. Moreover, we numerically identified lasing points and showed that the total energy diverges at these points, as illustrated in Figure \ref{fig:lasing}. With the characterisation of lasing points, we introduced an effective design approach which may find applications, for example, in energy-harvesting structures.

    Besides the possibility of energy gain and loss, we proved analytically in Theorem \ref{thm:Econs_Geigen} and numerically in Figure \ref{fig:Geigenvals} that the energy is conserved if the incident wave field's coefficients are precisely an eigenvector of the matrix $\mathcal{G}$ defined in \eqref{def:asymptotic_Si_tdep}. Eigenstates of the scattering operator correspond to a frequency balance where no frequency conversion takes place, hence resulting in energy conservation.

    Our results contribute a rigorous and flexible mathematical toolset for studying temporally modulated resonator systems, and open up new possibilities for designing non-reciprocal, tunable or gain-enhanced wave-control devices in acoustic and photonic contexts.

\section*{Code availability}
	The codes that were used to generate the results presented in this paper are available under \url{https://github.com/rueffl/tmod_energy}.

	\addcontentsline{toc}{chapter}{References}
	\renewcommand{\bibname}{References}
	\bibliography{refs}
	\bibliographystyle{plain}

\appendix
\section{Calculations}\label{app:calc}
    We shall present the mathematical details behind the derivation of the cross term in the total energy flux $P^\mathrm{tot}$, besides $P^{\mathrm{tr}}$ and $P^{\mathrm{ref}}$ as defined by \eqref{def:ptr} and \eqref{def:pref}, respectively. Recall that the total energy flux is given by \eqref{def:Fstatic} and upon substituting $u=u^{\mathrm{in}}+u^{\mathrm{ref}}$ we obtain
    \begin{align}\label{eq:P_appendix}
        P(x,t)=-\mathrm{Re}\left(\partial_t \overline{u^{\mathrm{in}}}\partial_xu^{\mathrm{in}}\right)-\mathrm{Re}\left(\partial_t \overline{u^{\mathrm{ref}}}\partial_xu^{\mathrm{ref}}\right)-\mathrm{Re}\left(\partial_t \overline{u^{\mathrm{in}}}\partial_xu^{\mathrm{ref}}\right)-\mathrm{Re}\left(\partial_t \overline{u^{\mathrm{ref}}}\partial_xu^{\mathrm{in}}\right),
    \end{align}
    valid for $x<x_1^-$. The first two terms in the above expression lead to $P^{\mathrm{in}}$ and $P^{\mathrm{ref}}$ upon averaging over time, while the second two terms yield $P^\mathrm{cross}$. We show now that the last two terms vanish and we do not see any contribution in the total energy flux. We substitute the expressions \eqref{def:inc_wave} and \eqref{def:ref_wave_nonasympt} and use the following notation:
    \begin{align}
        \omega_n:=\omega+n\Omega,\quad A_n:=\sum\limits_{m=-\infty}^\infty\left(S_{21}\right)_{nm}a_m,
    \end{align}
    where $\left(a_n\right)_{n=-\infty}^{\infty}$ are the coefficients of the incident wave field.
    Then we arrive at
    \begin{align}
        -\mathrm{Re}&\left(\partial_t \overline{u^{\mathrm{in}}}\partial_xu^{\mathrm{ref}}\right)-\mathrm{Re}\left(\partial_t \overline{u^{\mathrm{ref}}}\partial_xu^{\mathrm{in}}\right)\nonumber\\&=\mathrm{Re}\left(\sum\limits_{n=-\infty}^\infty\overline{a_n}\omega_n\mathrm{e}^{-\mathrm{i}\left(k^{(n)}x+\omega_nt\right)}\sum\limits_{m=-\infty}^\infty A_mk^{(m)}\mathrm{e}^{\mathrm{i}\left(-k^{(m)}x+\omega_mt\right)}\right)\nonumber\\&\qquad-\mathrm{Re}\left(\sum\limits_{n=-\infty}^\infty\overline{A_n}\omega_n\mathrm{e}^{-\mathrm{i}\left(-k^{(n)}x+\omega_nt\right)}\sum\limits_{m=-\infty}^\infty a_mk^{(m)}\mathrm{e}^{\mathrm{i}\left(k^{(m)}x+\omega_mt\right)}\right)\nonumber\\
        &=\mathrm{Re}\left(\mathrm{e}^{-2\mathrm{i}\frac{\omega}{v}x}\left(\sum\limits_{n=-\infty}^\infty\overline{a_n}\omega_n\mathrm{e}^{\mathrm{i}\frac{n\Omega}{v}\left(-x-vt\right)}\sum\limits_{m=-\infty}^\infty A_mk^{(m)}\mathrm{e}^{-\mathrm{i}\frac{m\Omega}{v}\left(x-vt\right)}\right)\right)\nonumber\\
        &\qquad-\mathrm{Re}\left(\mathrm{e}^{2\mathrm{i}\frac{\omega}{v}x}\left(\sum\limits_{n=-\infty}^\infty\overline{A_n}\omega_n\mathrm{e}^{\mathrm{i}\frac{n\Omega}{v}\left(x-vt\right)}\sum\limits_{m=-\infty}^\infty a_mk^{(m)}\mathrm{e}^{-\mathrm{i}\frac{m\Omega}{v}\left(-x-vt\right)}\right)\right).
    \end{align}
    Next, we average the above expression over time for which we introduce the following notation:
    \begin{align}
        &C_{nm}:=\frac{1}{v}\omega_n \omega_m\left(\overline{a_n} A_m \mathrm{e}^{-\mathrm{i}\frac{(2\omega + (n+m)\Omega)}{v}x} - \overline{A_n} a_m \mathrm{e}^{\mathrm{i}\frac{(2\omega + (n+m)\Omega)}{v}x}\right),\nonumber\\
        &J_{nm}:=\frac{1}{T}\int_0^T\mathrm{e}^{-\mathrm{i}(n-m)\Omega t}\,\mathrm{d}t=\begin{cases}
            1,&n=m,\\
            0,&n\neq m.
        \end{cases}
    \end{align}
    Then we have
    \begin{align}
        \frac{1}{T}\int_0^T-\mathrm{Re}&\left(\partial_t \overline{u^{\mathrm{tr}}}\partial_xu^{\mathrm{ref}}\right)-\mathrm{Re}\left(\partial_t \overline{u^{\mathrm{ref}}}\partial_xu^{\mathrm{tr}}\right)\,\mathrm{d}t=\sum\limits_{n=-\infty}^\infty\sum\limits_{m=-\infty}^\infty\mathrm{Re}\left(C_{nm}J_{nm}\right).
    \end{align}
    We note that $C_{mn}=-\overline{C_{nm}}$ for all $n,m\in\mathbb{Z}$ and conclude that
    \begin{align}
        \sum\limits_{n=-\infty}^\infty\sum\limits_{m=-\infty}^\infty\mathrm{Re}\left(C_{nm}J_{nm}\right)=\sum\limits_{n=-\infty}^{\infty}\mathrm{Re}\left(C_{nn}\right).
    \end{align}
    However, since $C_{nn}$ is purely imaginary, we get $\mathrm{Re}\left(C_{nn}\right)=0$.\par 
    With that we have proven that the mixed terms in the energy flux vanish, and hence do not contribute to the total flux. Thus, the total energy flux is indeed composed of the incident,  reflected and transmitted energy fluxes.    
\end{document}